\documentclass[10pt, conference, letterpaper]{IEEEtran}

\usepackage[utf8]{inputenc}

\usepackage{amsthm}
\usepackage{amsmath}
\usepackage{amsfonts}
\usepackage{nicefrac}

\newtheorem{definition}{Definition}
\newtheorem{theorem}{Theorem}

\usepackage[ruled, vlined, linesnumbered]{algorithm2e}

\usepackage{graphicx}
\usepackage{tikz}
\usepackage[caption=false]{subfig}

\usepackage{multirow}

\usepackage{url} 

\newcommand{\capacity}{b}
\newcommand{\fee}{fee}
\newcommand{\route}{\mathcal{R}}
\newcommand{\numPathTotal}{d}
\newcommand{\numPathNeeded}{k}

\title{The Merchant: Avoiding Payment Channel Depletion through Incentives}
\author{
\IEEEauthorblockN{Yuup van Engelshoven,  Stefanie Roos}
\IEEEauthorblockA{Delft University of Technology}
\IEEEauthorblockA{yuupvengelshoven@gmail.com, s.roos@tudelft.nl}
}

\begin{document}

\maketitle

\begin{abstract}
Payment channels networks drastically increase the throughput and hence scalability of blockchains by performing transactions \emph{off-chain}. In an off-chain payment, parties deposit coins in a channel and then perform transactions without invoking the global consensus mechanism of the blockchain.  However, the transaction value is limited by the capacity of the channel, i.e., the amount of funds available on a channel. These funds decrease when a transaction is sent and increase when a transaction is received on the channel. Recent research indicates that there is an imbalance between sending and receiving transactions, which leads to channel depletion in the sense that one of these operations becomes impossible over time due to the lack of available funds. 

We incentivize the balanced use of payment channels through fees. Whereas the current fee model depends solely on the transaction value, our fee policies encourage transactions that have a positive effect on the balance in a channel and discourage those that have a negative effect. This paper first defines necessary properties of fee strategies. Then, it introduces two novel fees strategies that provably satisfy all necessary properties. 
Our extensive simulation study reveals that these incentives increase the effectiveness of payments by $8\%$ to $19\%$.  
\end{abstract}

\section{Introduction}

Blockchains offer novel mechanisms for realizing secure non-custodial transactions~\cite{bitcoin}. However, due to the need for disseminating all transactions and blocks to all participants, basic blockchain designs do not scale in terms of throughput and latency~\cite{gervais2016security}. A promising solution to the lack of scalability in blockchains are off-chain transactions,  which allow for the execution of transactions without directly publishing them on the blockchain~\cite{gudgeon2019sok}.

The most mature approach to off-chain transactions are \emph{payment channel networks}~\cite{spilmanbitcoinj}. 
When creating a payment channel,  two nodes lock collateral on the blockchain that can then be spent during off-chain transactions between these two nodes. Transactions reduce the available funds of the sending node by the transaction amount while they increase the funds available to the recipient by the same amount. 
When two parties do not share a payment channel, they can forward the payment via multiple intermediaries. These intermediaries form one or several \emph{payment paths} along which the funds are reduced in one direction for all channels and increased in the other direction~\cite{lightning}. 

Recent work indicates that current routing protocols for finding payment paths predominantly use these bidirectional channels in one direction. As a consequence, channels become depleted in one direction, i.e., the available funds are low or 
completely gone. Channels without funds are no longer useful for completing payments~\cite{di2018routing,sivaraman2018routing}. Costly on-chain transactions, referred to as on-chain rebalancing~\cite{sivaraman2018routing}, that modify the deposits placed in the channel are required  to maintain liquidity. 

In this work, our goal is to incentivize the balanced use of channels to reduce depletion. 
Similar to previous approaches~\cite{di2018routing,sivaraman2018routing}, we focus on introducing fee functions for intermediaries in payments that incentivize maintaining balanced channels by promoting low fees for payment paths that avoid or reduce depletion. 

We are the first to motivate and formally define desirable properties of balance-incentive fee functions.  Based on these properties, we design two variants of a fee function that nodes apply locally to compute their fees. 
The key aspect of these fee functions is that the fee for a payment increases linearly with the degree to which it deteriorates the balance of a channel. The two variants differ in their treatment of payments that improve balances: The first variant still charges a small fee for such payments wheres the second variant does not charge fees. 

When conducting a payment, the sender then discovers several potential payment paths and chooses the path or paths for which the fees are lowest. As a consequence, the difference in communication overhead between the different fee functions is minimal as they all use the same routing algorithm and later locally decide which paths to utilize for payment.

We conducted a simulation-based evaluation for a variety of scenarios. Our fee function reduced the effect of depletion on the ratio of successfully conducted transactions by often more than $10\%$.
While channels still deplete over time, our work can be combined with on-chain rebalancing, which is required less frequently when combined with our fee function, and rebalancing loops~\cite{loopout}. 

Our routing scheme is fundamentally different to the source routing protocol currently implemented in Lightning~\cite{lightning}, Bitcoin's payment channel network. 
We point out two deployment options for integrating our work in the current Lightning algorithm. The first deployment option only requires minimal changes to Lightning's current protocol but increases the overhead as it requires three traversals of payment paths. The second option requires more significant changes but only requires two traversals as in the current Lightning implementation.

\section{System Model}
\label{sec:model}

In this section, we introduce our model of payment channel networks. While largely inspired by Bitcoin's Lightning network~\cite{lightning}, the model is more general in terms of fee functions and routing algorithms. 

\subsection{Payment Channel Network}

At time $t$, a payment channel network can be represented by a graph $G_t=(V_t,E_t)$ with nodes $V_t$ and channels $E_t \subset V_t \times V_t$, a non-negative balance function $\capacity_t: E_t \rightarrow  \mathbb{R}$, and a fee function $\fee_t: E_t \times \mathbb{R} \rightarrow \mathbb{R}$. 
 In other words, the balance function assigns each channel its available funds at time $t$ whereas the fee function gives the fee of the channel for a transaction value. 
 We assume that channels are always established in two directions, i.e., if $(u,v) \in E_t$, then $(v,u) \in E_t$. However, balances and fees usually differ between the two directions.
 
Two users $u,v \in V$ establish a channel with initial balances $c_{u,v}$ and $c_{v,u}$. In other words, $u$ can initially send up to $c_{u,v}$ coins to $v$ and $v$ can send up to $c_{v,u}$ coins to $u$. 
The \emph{total capacity} of the channel is $c_{u,v} + c_{v,u}$ where $c_{u,v}$ is the \emph{channel balance} in the direction from $u$ to $v$. 
Channel establishment requires a blockchain transaction. 

If $u$ transfers a positive amount $x \leq c_{u,v}$ to $v$, the balance from $u$ to $v$ changes to $c_{u,v}-x$ and the balance from $v$ and $u$ to $c_{v,u}+x$. Hence, the total capacity of the channel remains constant. Payment channel networks use cryptographic protocols as well as incentives to ensure that users cannot claim incorrect channel capacities without being detected or losing money~\cite{gudgeon2019sok}. 

At some point, $u$ and $v$ close the channel using another blockchain transaction. They receive the amount of funds corresponding to the balance of their channel direction, e.g., $u$ receives the balance in the direction from $u$ to $v$ at the time of the closure. 

As channel establishment is expensive and time-consuming, users do not open a channel unless they frequently transfer money between each other. For infrequent transactions to other nodes, they conduct a \emph{multi-hop payment}. 
In such a payment, the payment is forwarded via intermediaries from a source to a destination. In Lightning, the complete payment takes one single path~\cite{lightning} while other algorithms split the payment over multiple paths~\cite{roos2017settling}. 

In a multi-hop payment, the intermediaries charge fees according to the fee function $\fee$. Our model assumes that $v$ charges the fee for a channel $(v,u)$ when it is used from $v$ to $u$.  The assumption is reasonable when considering the impact of a payment on the balances: A payment from $v$ to $u$ decreases $v$'s balance and hence ability to send funds to $u$ but increases $u$'s balance. Thus, the transfer for this channel affects $v$'s liquidity negatively and hence incentives should target $v$ primarily. 
As a consequence, the sender $s$ does not pay a fee for the first channel in each path as $s$ itself is the first endpoint of that channel. 

Furthermore, this fee model is in line with Lightning's current fees. While Di Stasi et al.\ suggest to split the fee between both nodes adjacent to a channel, they do not provide any motivation  of this approach in terms of incentives~\cite{di2018routing}.

\subsection{Routing Algorithm}
\label{sec:routing}

A routing algorithm in payment channel networks discovers one or multiple paths to conduct a payment. In case of multiple paths, the algorithm furthermore determines how to split the payment over these paths. 
Informally, a routing algorithm has to find paths that do not violate the capacity restrictions of the payment network and can carry the complete payment value. 

As the lifetime of channels is in the order of hundreds of days\footnote{\url{https://1ml.com/statistics}}, the probability that the channel is closed during a transaction is ignored in our model. In other words, $V_t$ and $E_t$ are constant during a transaction but the balances and fees might change. 

Formally, let $\route$ be an algorithm that takes a payment network $G_t$, a source node $s$, a destination $d$, and a payment value $tx$ as input. Given that input, $\route$ computes a set $P$ of path-amount pairs, each of which consists of one path between $s$ and $d$ together with a transaction amount, and the total fee $f$ for conducting the payment. 

For the payment to be successful, it is required that
\begin{align}
\label{eq:txval}
\sum_{(p_i,val_i) \in P} val_i \geq tx+f.
\end{align} 
Most designs for multi-path routing~\cite{malavolta2017silentwhispers,roos2017settling} consider equality in Eq.~\ref{eq:txval}. Yet, a recent cryptographic protocol enables the source to send more than the total payment value and later reclaim any funds exceeding the transaction value. By sending more than the total transaction amount, the protocol accounts for the possibility of a partial payment failing~\cite{bagaria2020boomerang}. 

Furthermore, the payment is only successful if the capacity restrictions of the channels are respected. For simplicity, we assume that there are no concurrent payments. In the presence of concurrent payments using the same channels, previous work provides an approach for serializing such payments to prevent deadlocks~\cite{werman2018avoiding}. 

Note that the transaction value $val_i$ for the path $p_i$ consists of an amount $tx_i$ that ends up at the receiver and a fee $f_i$ for the intermediary nodes. Let $len(p_i)$ denote the length of the path, $e^{ij}$ the $j$-th channel on the path considered at time $t_{ij}$, and $f_{ij}$ the fee charged for using $e_{ij}$. The first channel on the path has to carry the full payment $val_i$. Each intermediary then keeps an amount corresponding to its fee, which is computed based on the value forwarded along the next channel, i.e., the actual transaction value plus the fees of its successors on the path. As a consequence, the fees are computed iteratively starting from the last channel $e_{il}$ with $l=len(p_i)$, so
\begin{align} 
\label{eq:feeIndividual}
f_{il} &= \fee_{t_{il}}\left(e_{il}, tx_i\right) \\
f_{ij} &= \fee_{t_{ij}}\left(e_{ij}, tx_i + \sum_{k=j+1}^{l} f_{ik}\right) \textnormal{ for } j=2, \ldots, l-1 \nonumber
\end{align}
Note that the index $j$ in Eq.~\ref{eq:feeIndividual} starts at 2 due to the first channel requiring no fees, as discussed above. The fee $f$ for the complete payment is then the sum of all partial fees, i.e.,
\begin{align}
\label{eq:feeTotal}
f = \sum_{(p_i, val_i) \in P} \sum_{j=2}^{len(p_i)} f_{ij}. 
\end{align}

Now, let $\tau$ denote the time at the start of the transaction. For a channel $e$, $cap_{\tau}(e)$ is the balance at the start of the transaction. Then, $\route$ has to guarantee that the sum of amounts forwarded via $e$ does not exceed the balance. Let $pos(p_i, e)$ denote the index corresponding to $e$'s position in the path $p_i$ if $e$ is contained in the path. Otherwise,  $pos(p_i, e)=0$. Formally, $\route$ has to adhere to the following constraints for all $e \in E_{\tau}$: 
\begin{align}
\label{eq:capConstraints}
\sum_{(p_i, val_i) \in P} \mathbf{1}_{(0,\infty)}& (pos(p_i, e)) \left(tx_i + \sum_{j=pos(p_i, e)+1} f_{ij} \right) \nonumber \\
&\leq cap_{\tau}(e) 
\end{align}
with $\mathbf{1}_{S}(x)$ being $1$ if $x$ is in the set $S$ and $0$ otherwise. 

If $\route$ fails to find a solution that satisfies Eqs.~\ref{eq:txval} and~\ref{eq:capConstraints}, $\route$ responses with a routing failure.  Otherwise, the process can advance to the payment phase. 

\subsection{Security of Multi-hop Payments}
When conducting a payment along the paths discovered by $\route$, sender $s$ and destination $d$ require no counter party risk and atomicity~\cite{gudgeon2019sok}. 
No counterparty risk means that none of the involved parties can violate their obligations to conduct a transfer as promised without repercussions. 
Atomicity means that all transfers along individual channels have to succeed for the payment to succeed. Otherwise, if only one of the transfers fails, $s$ can reclaim the total amount sent, including the fees. 

Payment channel networks achieve no counterparty risk and atomicity through so-called locks. Locking a certain payment amount in a channel corresponds to a commitment to utilize the amount for this payment. Only after all nodes involved in the payment have signed the necessary commitments, the actual payment is conducted. If the payment does not happen, the lock expires after a time-out and the node can use the funds for other payments. 
If a party defaults their obligations after locking the funds, the party loses money as parties typically have to pay before they are paid.  
There exist a variety of locking mechanisms~\cite{htlcs,tairi2019a2l,sprites,malavoltamulti}, differing with regard to performance, privacy, and additional security goals. We discuss their suitability for our routing protocol in Section~\ref{sec:deployment}.

\section{Fee Function Properties}
\label{sec:fee}

Before introducing our fee function, this section narrows down the type of fee functions to functions that incentivize balanced channels. 

Generally, let the \emph{reference point} $ref(e)$ of a channel $e$ between $u$ and $v$ be the capacity that $u$ and $v$ consider optimal. In previous studies~\cite{di2018routing}, the reference point corresponds to splitting the total capacity  equally between the two directions, i.e., $ref(e)=\frac{\capacity_t((u,v))+\capacity_t((v,u))}{2}$. However, in practice, $u$ and $v$ might preferably use the channel in one direction and hence consider an unequal distribution of the total capacity  more beneficial.
By allowing nodes to define their own reference point, this model provides the flexibility necessary to incorporate heterogeneous financial relationships.  

In Section~\ref{sec:model}, the fee function $\fee_t$ takes the channel and the transaction value as input. 
As our fee function should incentivize balanced channels, the relevant aspects of the channel are the current balance of the channel in each direction and its reference point. 
Thus, we have a fee function with four inputs. 

While all inputs are non-negative real numbers, there are some restrictions on the choice of the values for payments to be feasible. For the channel to be useful and worth establishing, the total capacity needs to exceed 0. For the payment to succeed, the transaction value can be at most equal to the balance in the respective direction. Similarly, the reference point needs to be an achievable balance, i.e., between 0 and the total capacity. Thus, our input set is 
\begin{align}
\label{eq:input}
&I = \\
 \{(c_{-}, c_{+}, r, x) \in \mathbb{R}_{\geq 0} &: c_{-}+c_{+} > 0, x \leq c_{-}, r \leq c_{-}+c_{+}\} \nonumber
\end{align}
with $\mathbb{R}_{\geq 0}$ denoting the set of non-negative real numbers. 
So, our designed fee function has the form $\fee_{bal} \colon I \rightarrow \mathbb{R}$ with the four inputs being the two balances, the reference point, and the transaction value. 

Our first and key property ensures that the fee function indeed incentivizes moving towards the reference point.
\begin{definition}
\label{def:balInv}
A fee function $\fee_{bal} \colon I \rightarrow \mathbb{R}$ is \emph{balance-incentive} if 
\begin{align}
\label{eq:balInv}
&\forall c_{-}, c_{+},r,x_1,x_2 \in \mathbb{R},x_1\leq c_{-}, x_2\leq c_- \colon \nonumber \\
& \fee_{bal}(c_{-}, c_{+},r, x_1) \leq \fee_{bal}(c_{-}, c_{+},r,x_2) \\
& \implies |c_{-}-x_1-r| \leq |c_{-}-x_2-r|. \nonumber 
\end{align}
\end{definition}
In other words, the fee function increases with $|c_{-}-x-r|$. We say that the balance of a channel deteriorates if it moves further away from the reference point. If it moves towards the reference point, the balance improves. 

Second, the fee function should be \emph{continuous} in all variables so that similar values result in similar fees. 

Third, splitting a payment into a large number of smaller payments and transferring all payments along the same channel creates a high overhead. Thus, the fee function should disincentivize such splits. In other words, the fee when transferring the complete value is at most as high as the sum of the individual fees when splitting the value.  Formally, this implies that for all $c_{-},c_{+}, r,x_1,x_2$ with $x_1 + x_2 \leq c_{-}$, we have 
\begin{align}
\begin{split}
 &\fee_{bal}(c_{-},c_{+},r,x_1+x_2) \leq \\
&\fee_{bal}(c_{-},c_{+},r, x_1) + \fee_{bal}(c_{-}-x_1,c_{+}+x_1, r, x_2).
\end{split}
\end{align}
In other words, the fee function is \emph{subadditive in the transaction value}.

The above three properties --- balance-incentive, continuous, and subadditive--- form the basis for our design of a fee function. 

\section{The Merchant}
\label{sec:merchant}

This section describes The Merchant, which consists of i) a method for integrating fees into a multi-hop routing algorithm $\route$ and  ii) a fee function. The Merchant can use one of two variants for the fee function. 

\subsection{Routing Algorithm Integration}
\label{sec:routeIntegration}

The Merchant provides a method for integrating a preference for cheap paths into routing that is independent of the fee function. 
Assume that the routing algorithm $\route$, as defined in Section~\ref{sec:routing}, can be parametrized such that it returns $\numPathTotal$ paths and splits the transaction value equally over all paths. 
The Merchant uses the cheapest $\numPathNeeded \leq \numPathTotal$ for the payment. 

More precisely, to conduct a transaction of value $val$, $\route$ looks for paths for a total transaction value of $val/\numPathNeeded \cdot \numPathTotal$.   If $\route$ finds $m \geq \numPathNeeded$ paths, routing is successful. The sender then selects the $\numPathNeeded$ paths with the lowest fees. 
We discuss how to realize such a path selection securely when exploring deployment options in Section~\ref{sec:deployment}. 

\subsection{Fee Function}
\label{sec:feefunc}

Our strategy is to have a fee that is proportional to the degree a transaction changes the balance with regard to the reference point. Thus, for a factor F, a straight-forward choice for the fee function is
\begin{align}
\label{eq:ourfee}
\fee_{bal}(c_{-},c_{+}, r, x) = \left(1 + \frac{|c_{-}-x-r|-|c_{-}-r|}{c_{-}+c_{+}}\right)\cdot F. 
\end{align}
If the transaction results in the balance getting closer to the reference point, the fraction in Eq.~\ref{eq:ourfee} is negative, thus leading to a low fee. Otherwise, if the transaction results in the balance being further from the reference point, the fraction is positive and hence increases the fee. 

\begin{theorem}
\label{thm:feebal}
The fee function $\fee_{bal}: I \rightarrow \mathbb{R}$, as defined in Eq.~\ref{eq:ourfee}, is balance-incentive, continuous, and subadditive. 
\end{theorem}
\begin{proof}
By design, $\fee_{bal}$ is balance-incentive as it is increasing in $|c_{-}-x-r|$. 
As the input set $I$ excludes $c_-+c_+=0$, $\fee_{bal}$ is also continuous. 

For subadditivity, consider that
\begin{align*}
&fee_{bal}(c_{-},c_{+}, r, x_1) + \fee_{bal}(c_{-}-x_1,c_{+}+x_1, r, x_2) \\
=& \left(1 + \frac{|c_{-}-x_1-x_2-r|-|c_{-}-x_1-r|}{c_{-}+c_{+}} \right. \\
&\left. +1+\frac{|c_{-}-x_1-r|-|c_{-}-r|}{c_{-}+c_{+}}\right)\cdot F \\
=& (2 + \frac{|c_{-}-x_1-x_2-r|-|c_{-}-r|}{c_{-}+c_{+}})\cdot F \\
> & (1 + \frac{|c_{-}-x_1-x_2-r|-|c_{-}-r|}{c_{-}+c_{+}})\cdot F \\
=& \fee_{bal}(c_{-}, c_{+}, r, x_1+x_2). 
\end{align*} 
Hence, $\fee_{bal}$ is subsadditive. 
\end{proof}

\subsection{Variant: No Fees For Improved Balance}
The above fee function adds a fee for each channel on  a path, excluding the first one. However, it is debatable if changes to balances that improve the ability of a node to route future payments should entail a fee. 
Consider a path of length four of which the balance of two channels worsens and the balance of two channel improves. 
As fees are charged even for channels with improved balance, the sender still prefers a path of length two with two channels for which the balance deteriorates similarly to the two channels with deteriorated balances on the path of length four. 
While preferring shorter paths has the advantage of lower latencies and less locked collateral on the path, the decision not to choose the path that offers some improvement in terms of balances seems disadvantageous in terms of preventing depletion. 

The obvious option to counteract such issues is to allow for negative fees. However, it seems unlikely that nodes will agree to pay others to use their channels in exchange for potential but not certain future gains. 
In particular, nodes have to lock the collateral they promised for some time period, meaning it is not available in either direction for this period. 
Hence, we decided not to use negative fees. Rather, whenever the capacity changes in a positive manner, the node does not charge a fee.

Precisely, the variant of our fee function is 
\begin{align}
\label{eq:ourfee0}
\begin{split}
&\fee^0_{bal}(c_{-},c_{+}, r, x) = \\
& \begin{cases}
0, &\text{ if }  |c_{-}-x-r|\leq |c_{-}-r|\\ 
\frac{|c_{-}-x-r|-|c_{-}-r|}{c_{-}+c_{+}}\cdot F, & \text{ otherwise }
\end{cases}
\end{split}
\end{align}

\begin{theorem}
\label{thm:feebal0}
The fee function $\fee^0_{bal} \colon I \rightarrow \mathbb{R}$, as defined in Eq.~\ref{eq:ourfee0}, is balance-incentive, continuous, and subadditive. 
\end{theorem}
\begin{proof}
By design, $\fee^0_{bal}$ is balance-incentive as it is increasing in $|c_{-}-x-r|$. Furthermore, it is continuous in all variables for inputs in $I$. 

For subadditivity, we only have to consider the case $\fee^0_{bal}(c_{-},c_{+}, r, x_1+x_2)>0$. 
If both $\fee^0_{bal}(c_{-},c_{+}, r, x_1)>0$ and $\fee^0_{bal}(c_{-}-x_1,c_{+}+x_1, r, x_2)>0$, analogously to the proof of Theorem~\ref{thm:feebal}, 
\begin{align*}
&fee^0_{bal}(c_{-},c_{+}, r, x_1) + \fee^0_{bal}(c_{-}-x_1,c_{+}+x_1, r, x_2) \\
&=\frac{|c_{-}-x_1-x_2-r|-|c_{-}-r|}{c_{-}+c_{+}}\cdot F
\end{align*}
 and hence $\fee^0_{bal}$ is subadditive. 
Otherwise, if $\fee^0_{bal}(c_{-},c_{+}, r, x_1)=0$, then $|c_{-}-x_1-r| \leq |c_{-}-r|$ and hence 
\begin{align*}
&\fee^0_{bal}(c_{-},c_{+}, r, x_1+x_2)\\
=&\frac{|c_{-}-x_1-x_2-r| -|c_{-}-r|}{c_{-}+c_{+}}F \\
\leq &\frac{|c_{-}-x_1-x_2-r| -|c_{-}-x_1-r|}{c_{-}+c_{+}}F \\
=&\fee^0_{bal}(c_{-}-x_1,c_{+}+x_1, r, x_2) +\fee^0_{bal}(c_{-},c_{+}, r, x_1)
\end{align*}
Thus, $\fee^0_{bal}$ is subadditive for the this case as well.
Similarly, if $\fee^0_{bal}(c_{-}-x_1,c_{+}+x_1, r, x_2)>0$, then $|c_{-}-x_1-x_2-r| \leq |c_{-}-x_1-r|$. Subadditivity follows as above from
\begin{align*}
|c_{-}-x_1-x_2-r| - |c_{-}-r| \leq |c_{-}-x_1-r| - |c_{-}-r|. 
\end{align*}
So, $\fee^0_{bal}$ is subadditive. 
\end{proof}

\section{Evaluation}

In this section, we evaluate The Merchant with respect to its impact on success ratio and communication overhead. For this purpose, we first choose a routing algorithm that is suitable for integration. Then, we present our simulation-based evaluation, which includes a comparison to state-of-the-art fee functions for payment channel networks. 

\subsection{Routing Algorithm}
Lightning's current routing algorithm only discovers one path and is  hence unsuitable for our design~\cite{lightning}. While there exist proposals for how to realize secure multi-hop payments in Lightning, there is no deployment. Previous research applied methods for discovering node-disjoint paths using source routing~\cite{di2018routing}.
However, there are few node-disjoint paths in Lightning~\cite{rohrer2019discharged}, indicating that such an algorithm frequently only considers one path option. Thus, requiring node-disjoint paths seems unrealistic given our knowledge of the network topology. 

Hence, this evaluation focuses on SpeedyMurmurs~\cite{roos2017settling}, an alternative routing algorithm that supports multiple paths by default. In contrast to Lightning, SpeedyMurmurs does not require the complete topology information and is hence scalable. If the Lightning network continues to grow in size, it will likely integrate such an algorithm that can route without a global view. 

In a nutshell, SpeedyMurmurs constructs $\numPathTotal$ spanning trees and for each tree assigns node coordinates based the nodes' positions in the tree. In each tree, nodes then discover a path in a distributed manner by forwarding to a neighbor that is closer to the destination in terms of the coordinates and has a channel of sufficient balance. 
One of the key drawbacks of SpeedyMurmurs is that failures in only one tree result in a failure of the complete payment. Only requiring $\numPathNeeded$ to succeed should improve the performance drastically in addition to the increase due to more suitable fee functions. 

Integrating The Merchant into SpeedyMurmurs requires adding redundancy, as detailed in Section~\ref{sec:routeIntegration}, and fee functions, as detailed in Section~\ref{sec:feefunc}. 

For comparison to the current Lightning implementation, our simulation also included an algorithm that only selects one cheapest path at the source. More concretely, as in Lightning, we assume that the total capacity of a channel is publicly known but not the balances. In our implementation of the cheapest-path algorithm, the source node executed Dijkstra locally using its snapshot of the topology, the known fees, and the known total capacities. The cost of each channel corresponded to its fee, which were fixed as in Lightning. The search algorithm started from the receiver such that the fees could be computed. During the execution, channels without the necessary total capacity were excluded. Note that this could still lead to failures due to insufficient balances. 

The above source routing algorithm incorporates the key idea of Lightning's routing: choose the cheapest path while excluding channels that clearly entail failures. However, the cost functions used in deployed Lightning clients slightly differ in terms of the cost function. Each of the three client implementation change the cost function in a different manner, adding preferences for short paths, low delays, or channels that previously led to successful payments~\cite{tochner2019hijacking}. For simplicity, we hence focused on the key idea of the algorithm and only discuss potential impacts of modifications while presenting the results. 

\subsection{Simulation Model}

Our simulation framework relies on the simulator written for SpeedyMurmurs\footnote{\url{https://crysp.uwaterloo.ca/software/speedymurmurs/}}, which we extended to include redundancy and fees.

As the cryptographic operations differ between implementations of payment channel networks, the simulator does not include the cryptographic overhead and focuses on the path discovery. As there are methods for serializing concurrent transactions~\cite{werman2018avoiding}, there is also no concurrency implemented.  Instead of concrete latencies, the hop count is used as a measure of transaction duration. In such a setting, the best measure of time is the number of transactions.

\begin{figure*}
     \centering
     \subfloat[x=0.05]{\includegraphics[width=0.49\textwidth]{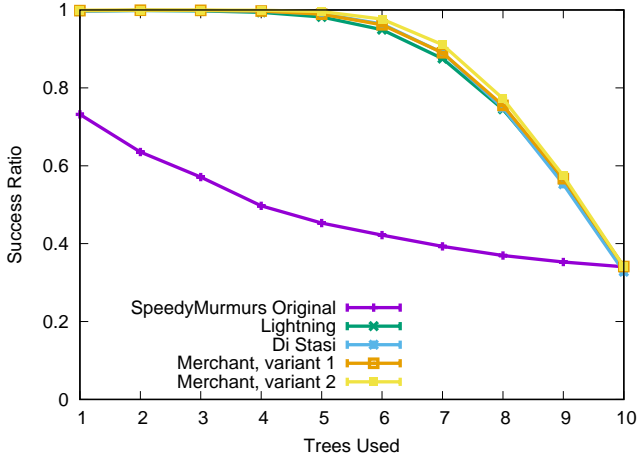}\label{fig:5p}}\hfill 
     \subfloat[x=0.5]{\includegraphics[width=0.49\textwidth]{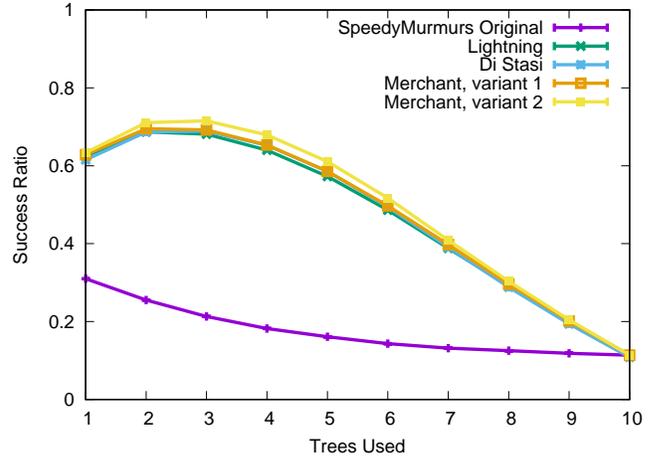}\label{fig:50p}}
     \caption{Success ratio averaged over 100,000 transactions for different fee functions (BA graph)}
     \label{fig:succAll}
\end{figure*}

\begin{figure*}
     \centering
     \subfloat[x=0.05]{\includegraphics[width=0.49\textwidth]{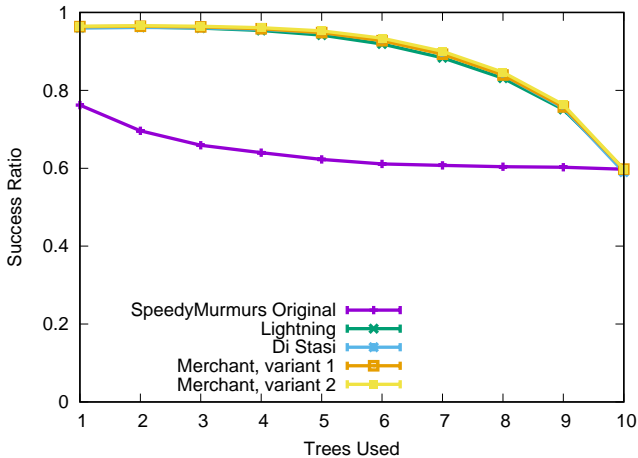}\label{fig:5pLN}}\hfill 
     \subfloat[x=0.5]{\includegraphics[width=0.49\textwidth]{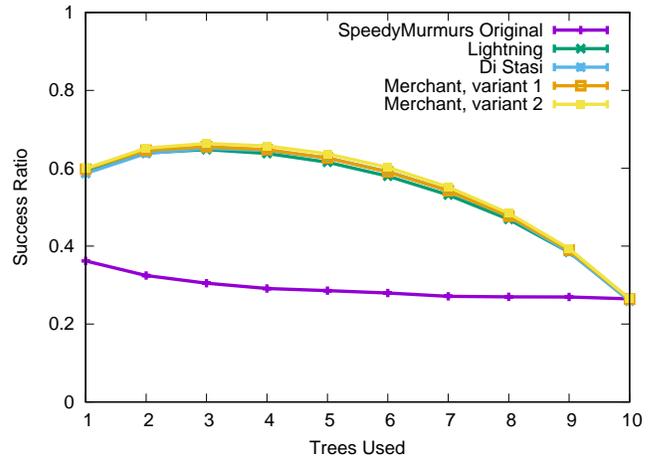}\label{fig:50pLN}}
     \caption{Success ratio averaged over 100,000 transactions for different fee functions (LN graph)}
     \label{fig:succAllLN}
\end{figure*}

\subsection{Simulation Setup}
For comparison, SpeedyMurmurs was simulated i) in its original form, ii) with redundancy and state-of-the-art fee functions, and iii) with The Merchant. For original SpeedyMurmurs, we varied $\numPathTotal$ between 1 and 10, when using redundancy we chose $\numPathTotal=10$ and varied $\numPathNeeded$ between 1 and 10. 

Transaction and fee values were in satoshi, i.e., one millions of Bitcoins. 
Lightning uses a constant base fee $b$ and a fee rate $r$ such that its total fees for a transaction amount $val$ are $b+r\cdot val$ for a channel. Our simulation used the default values of $b=1$ and $r=10^{-6}$ satoshi.  
In addition to our fee function, the study considered a fee function suggested in previous work. Concretely, di Stasi et al.\ designed a modification of Lightning fees that increases the fee rate for the partial amount of the payment that deteriorates the balance.
As suggested in the respective paper~\cite{di2018routing}, we chose the two rates to be $0.01$ and $0.03$ satoshi 
for non-deteriorating and deteriorating (partial) payments. We refer to the fee function as  di Stasi. 
For The Merchant, we set $F=1$, and the reference point was always $50\%$ of the total capacity. 

Our models for the topology and initial balances incorporate statistics about the Lightning network~\footnote{\url{https://1ml.com/statistics}}. 
The topologies were a real-world Lightning snapshot and a synthetic Barabasi-Albert graph.
The snapshot was from  March 1, 2020, snapshot 04\_00\footnote{\url{https://gitlab.tu-berlin.de/rohrer/discharged-pc-data}}. It contains $6329$ nodes with the average number of channels being $10.31$. 
As Lightning is a small-world network~\cite{rohrer2019discharged}, we also modeled the topology as a Barabasi-Albert graph of 5000 nodes with an average degree of roughly 10. 
The initial balances had an expected value of $init=2400000$ satoshi. While the statistics for the Lightning network indicate an exponential distribution, we also considered normally distributed initial balances to study the impact of the balance distribution. 
In the absence of a reliable data set to model transactions, source and destination pairs were chosen uniformly at random. The transaction value was either exponentially or normally distributed with the expected value being $x*init$ for $x \in \{0.01, 0.05, 0.25, 0.5, 1\}$. Each simulation run considered 100,000 transactions. 
Results were averaged over 20 runs.  

\subsection{Results}

This section first discussed the key performance evaluation criterion, the success ratio, followed by short discussions of the impact of the balance distribution, the overhead, and a comparison to cheapest-path source routing.

\textbf{Success ratio:} Figure~\ref{fig:succAll} displays the success ratio for the Barabasi-Albert graph. Initial balances and transaction values were exponentially distributed, the value $x$ determining the expected transaction amount was either $x=0.05$ (in Figure~\ref{fig:5p}) or $x=0.5$ (in Figure~\ref{fig:50p}), corresponding to low- and high-value transactions. 
Results are displayed with standard deviation but the differences between runs were minimal, so that the error bars are barely visible. 
In both cases, redundancy greatly improved the success ratio in comparison to the original SpeedyMurmurs due the possibility of overcoming failures. 
Balance-incentive fee functions further increased the success ratio. Di Stasi  fees and The Merchant variant 1 have a very similar effect on the success ratio. As a consequence, the respective lines are hard to distinguish. 
The second The Merchant variant that allowed for zero fees achieved the highest success ratio. Indeed, the increase is close to $8\%$ in comparison to the original Lightning fees.

The success ratio was highest if $\numPathNeeded$ was relatively low, i.e., 5 or below. For large payments, which were more common for $x=0.5$, it was often not feasible to forward a payment via only one path. As a consequence, the maximal success ratio was achieved for $2$ or $3$ trees.  
\begin{figure}
     \centering
     \includegraphics[width=0.8\linewidth]{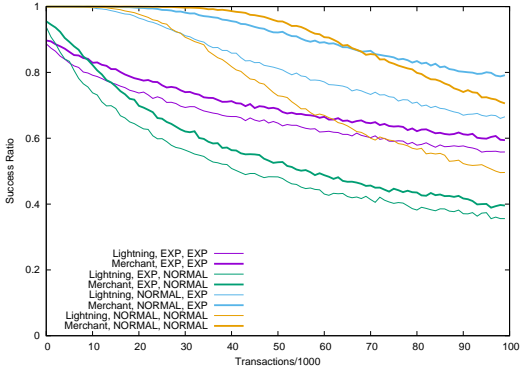}
     \caption{Success ratio over time for $x=0.25$, the first element of the labels indicates the fee function, the second indicates the type of initial balance distribution (EXP=exponential, NORMAL=normal), and the third the type of transaction distribution (Barabasi-Albert, 5000 nodes)}
     \label{fig:time}
\end{figure}
\begin{figure*}
     \centering
     \subfloat[x=0.05]{\includegraphics[width=0.49\textwidth]{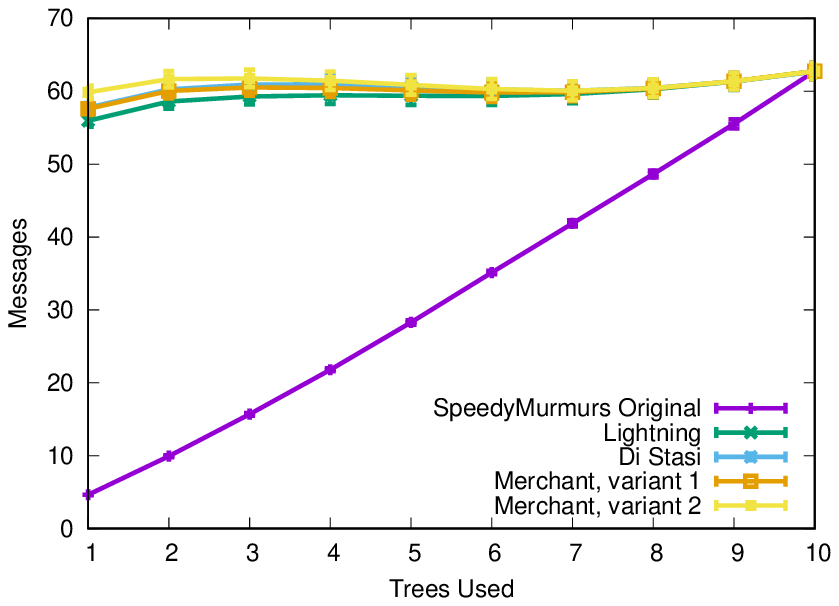}\label{fig:5p-hop}}\hfill 
     \subfloat[x=0.5]{\includegraphics[width=0.49\textwidth]{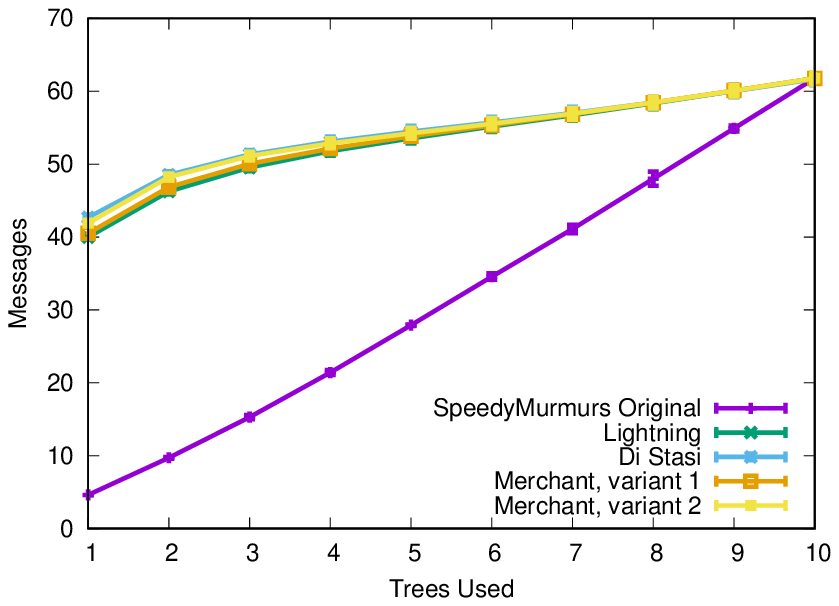}\label{fig:50p-hop}}
     \caption{Communication overhead in number of messages averaged over 100,000 transactions for different fee functions (Barabasi-Albert, 5000 nodes)}
     \label{fig:hopAll}
\end{figure*}

The results are similar for the Lightning snapshot as Figure~\ref{fig:succAllLN} indicates. The Merchant variant 2 consistently achieves the highest success ratio. The Merchant variant 1 and Di Stasi perform almost identical, so that the lines cannot be distinguished. The positive impact of The Merchant was less for Lightning than Barabasi-Albert. The reason for this lower increase in success lies in the low number of distinct path in Lightning. Lightning has a hub-and-spoke topology~\cite{rohrer2019discharged}, meaning that all paths between two nodes likely share at least one hub. 
Thus, it is impossible to avoid the use of certain channels and different fees for the channel have less impact.

In order to assess the level of depletion over time, Figure~\ref{fig:time} shows the success ratio for batches of 1000 transactions when applying Lightning or The Merchant variant 2 fees. A balance-incentive fee function increased the success ratio almost from the start. After 100,000 transactions, The Merchant's success ratio was always at least 10\% higher than for Lightning fees. However, all algorithms experienced a pronounced drop in success over time. The results did not indicate a monotonous decrease, however, the oscillations were not significant.

\textbf{Balance distribution:} Furthermore,  Figure~\ref{fig:time} displays differences between exponentially and normally distributed initial balances for the Barabasi-Albert graph. The positive effect of The Merchant is more pronounced for the latter.
 Indeed, after 100,000 transactions, the difference in success ratio between Lightning fees and The Merchant was up to $19\%$ for normally distributed balances and exponentially distributed transaction amounts. 
For exponentially distributed balances, few channels had a high enough capacity to route payments of a medium or high value. Thus, there were few options for paths. In the absence of choices, the same path was taken regardless of the fees. 
For normally distributed balances, channel capacities had less variance resulting in more options and hence The Merchant's fees could actually lead to different paths than Lightning fees.

\textbf{Overhead:} Redundancy clearly increased the communication overhead of successful payments in terms of number of forwarding operations, as shown in Figure~\ref{fig:hopAll} for the same parameters as in Figure~\ref{fig:succAll}. The differences between the fee functions were small but significant. Note that our metric only considered the routing process, which also included paths that were not chosen for payment. Hence, the differences were not due to the varying degree of preference for shorter paths.  Rather, as the depletion progressed differently, balance-incentive fee functions like The Merchant had more path options available and hence less paths that terminated after the first few hops. As a consequence, the overhead was slightly higher.  

\textbf{Comparison to source routing:} On the Barabasi-Albert graph, using source routing along one cheapest path incurred success ratios of approximately 89.5\% and 47.7\% for $x=0.05$ and $x=0.5$, respectively. For the Lightning topology, the corresponding success ratio was 88.4\% and 46.0\%. 
As displayed in Figures~\ref{fig:succAll} and~\ref{fig:succAllLN}, source routing had a considerable lower success ratio than multi-path routing that allowed the selection of a subset of paths.   
As stated above, our implementation of cheapest-path source routing is not an exact implementation of Lightning, which implements a number of variants. However, as our simulation relied on default fees for all nodes, the cheapest path was also a shortest path. Intuitively, the probability of failure should be least when using shortest paths as the number of channels that can have insufficient balance is lowest. Thus, it seems unlikely that any Lightning implementation can increase the success ratio for one routing attempt. Clearly, multiple routing attempts will increase the success ratio. Yet, further routing attempts in payment channel networks are only possible after the smart contracts for conditional payments expire, which takes hours to days~\footnote{\url{https://github.com/lightningnetwork/lightning-rfc/blob/master/07-routing-gossip.md#the-channel_update-message}}. As a consequence, we did not consider multiple attempts a viable option. 

In summary, The Merchant indeed reduces the effect of depletion at a barely noticeably increased communication overhead in comparison to other fee functions.
However, on its own, the algorithm does not solve the issue of depletion. Note that depletion can often not be fully solved by fee incentives alone: If funds are only utilized in one direction, e.g., with a customer paying for a service, having incentives for the other direction does not have any impact. However, we have seen that balance-incentive fee functions can mitigate depletion.  

\section{Deployment Options}
\label{sec:deployment}

Our protocol differs from the assumptions in current deployment in key aspects. Hence, it seems non-trivial to integrate it into existing systems such as Lightning.
In the following, we first describe the steps of Lightning routing and its key differences to The Merchant. Afterwards, we propose and compare two deployment options for our protocol.

\subsection*{Lightning routing}
In Lightning~\cite{lightning}, the source first determines one routing path. The available public information includes the topology, the public keys of all nodes, the base fee and fee rate, as well as the total capacity of the channel, i.e., the sum of the collateral locked by the two parties in a channel.
The exact manner of selecting the path differs between the three available Lightning clients: Lightning Labs' lnd\footnote{\url{https://lightning.engineering/}}, ACINQ's eclair\footnote{\url{https://acinq.co/}}, and Blockstream's c-lightning\footnote{\url{https://blockstream.com/lightning/}}. Yet, all three aim to select a path that is cheap, short, and likely to lead to a successful payment. 
After having decided on a path, the source encrypts the path information such that each node removes a layer of encryption that reveals the next hop on the path but not any further hops~\cite{danezis2009sphinx}. Not explicitly revealing more path information than necessary should protect the anonymity of sender and receiver but multiple studies have raised questions regarding the effectiveness of this approach in the presence of metadata like payment value and timings~\cite{kappos2020empirical,nisslmueller2020toward}.

Then, the sender starts building the path. Each pair of nodes along the path commits to the payment. If one node does not commit, the payment fails. Reasons not to commit could be a lack of available funds as the sender only knows the maximal amount of funds available and the actual amount can be much lower. However, considering the payment not beneficial enough or being offline are also reasons. 
If all nodes commit to the payment, the payment is executed.  
Note that the current single-path algorithm can be modified to choose multiple paths~\cite{atomicmultipath}. 

Our approach seems fundamentally different as the routing and fee decisions are made by nodes locally depending on information about the current state of the channel, which is unavailable to the source. Hence, we cannot simply replace the first step of the protocol where the source decides on a path locally with our path finding algorithm. 
Without knowing the path or paths, the source cannot encrypt and determine the total value, i.e., the transaction value including fees, for the path. Yet, the commitments to pay a certain value need to include the fee. 
Last, our protocol does not actually utilize all discovered paths but only the cheapest ones. 

\subsection*{Pre-Routing}
The first option to integrate our protocol is to replace source routing with distributed path discovery as follows: The receiver initiates the routing for the sender's coordinates.
In this manner, the sender obtains all the paths and the total fee per path if the routing reaches the sender. The intermediaries keep note of the fees they indicate for this route construction, e.g., by associating it with a query ID.
Once the sender has the complete path and fee information included in a successful routing query, it can complete the remaining payment steps as in the current Lightning version, albeit with multiple paths. Initially querying for more paths than the selected ones has no impact as nodes do not make any commitments during this initial \emph{pre-routing} phase.

The most obvious difference to the current protocol is the higher communication overhead due to the first step involving a distributed path selection. In Lightning, the routing includes two traversals of the path: one to make the commitments, one to execute the payments. Pre-routing adds a third traversal for $\numPathTotal$ paths while commitment and payment phase only use $\numPathNeeded$ paths. 

A second point for discussion is that intermediaries might be unwilling to forward the payment for the previously stated fee if their local situation changes between pre-routing and the actual commitments. For instance, the balances can change such that they now should charge a higher fee or the channel does no longer have the necessary balance. However, as stated above, failures due to intermediaries not forwarding the payment are common in the current network as well. Indeed, in our version, the source at least knows that the payment was possible very recently. The exact time between the pre-routing and the actual commitment depends on path length and latencies but it is unlikely to exceed a few seconds. In contrast, the current protocol does not have recent information about the channel state and hence a higher probability to choose a channel that lacks the necessary balance. 

The last difference relates to privacy, specifically sender anonymity. The pre-routing requires sender coordinates. There exist algorithms for generating anonymized versions of these coordinates but they still reveal some information about the location of the sender in the network~\cite{roos2016anonymous}. Given the issues with anonymity in the current implementations of Lightning, it remains unclear if the information revealed during pre-routing decreases the privacy further. In general, improving the anonymity of payment channel network routing is an important area of future work for both Lightning and The Merchant. Note that the path information can be relayed to the source in an encrypted manner and hence does not reveal the path to intermediaries. 

In summary, we have a feasible option for integrating our protocol into Lightning while maintaining large parts of their current payment process. The key disadvantages of this deployment option is an increased communication overhead. 

\subsection*{Local Routing Decisions}
Our second idea is to combine the path selection with the commitments, i.e., make commitments to pay while also deciding on the path. The paths are then determined locally by forwarding nodes. In this manner, we can get rid of the separate pre-routing step and hence reduce latency and communication overhead. Furthermore, nodes commit directly rather than only indicating a potential cooperation during pre-routing, hence reducing opportunities for failures. Privacy-wise, the routing in this second option requires the receiver coordinates for finding the path, hence the problem is similar to the solution with pre-routing. 

A recent work presents a method for securely realizing local routing protocols like the one described above~\cite{eckey2020splitting}. In particular, it allows nodes to further split the payment, which we can utilize to realize multiple paths. 
The authors show that their protocol works in Lightning. 

The presented protocol has two aspects that make it not immediately applicable to our scheme. First, it assumes atomic payments, i.e., all of the partial payments are executed whereas our protocol only executes a subset of payments along the cheapest paths. Second, the authors leave it open how to integrate fees into their algorithm. 

We think that these drawbacks can be overcome by combining local routing with Boomerang~\cite{bagaria2020boomerang}. Boomerang allows the source to route more coins than the payment value and then reclaim the left-over coins. Thus, the source can add a maximal value it is willing to pay as a fee to the routed payment value. If the payment succeeds, the source can reclaim both funds routed along paths that have not been selected as cheapest paths and any remaining fees. 

This deployment option has a lower overhead and less latency as it avoids the extra pre-routing step but is considerably harder to integrate into Lightning as it drastically changes the routing process. 

This section shows that we can integrate our protocol into Lightning while maintaining the security guarantees in at least two ways.

\section{Related Work}

A large fraction of the work on multi-hop payment channels focuses on the construction of cryptographic protocols for enabling no counterparty risk, atomicity, and privacy properties such as anonymity~\cite{malavoltamulti,htlcs,sprites,egger2018amcu}.

Furthermore, there are protocols that allow for redundancy~\cite{bagaria2020boomerang} and deadlock prevention~\cite{werman2018avoiding}, which are applicable for a wide range of routing algorithms. Our design incorporates redundancy but we combine it with a fee mechanism for mitigating depletion. 

With regard to routing payments, early designs~\cite{lightning,Flare2016,roos2017settling} do not consider the aspect of depletion. They mostly evaluate their algorithm in a static setting, i.e., without changing balances.
The static setting allows for a comparison to an optimal algorithm in terms of success ratio, as the maximal payment value then corresponds to the maximum flow~\cite{roos2017settling}.  

The two approaches closest to our work are Spider~\cite{sivaraman2018routing} and a fee function designed by di Stasi et al.~\cite{di2018routing}. 

Spider proposes a packet-switched network and divides each payment into units. The source routes each payment individually using a flow model that takes the current balances into account. 
However, Spider assumes that all current balances are known to the sender, which is clearly unrealistic in a large network with frequent changes.  A recent  extension  to the original design hence instead routes only along pre-selected paths and only probes the balances of these paths. However, the probing process leads to delays in adjusting to balance changes. Furthermore, the path discovery still requires global knowledge of the network topology and the implementation only scales to about 100 nodes~\cite{sivaraman2018high}. 

In contrast, di Stasi et al. propose to change the fee function in Lightning. In Lightning, nodes charge a constant base fee $b$ and a fee linear in the transaction amount~\cite{lightning}. 
The proposed fee strategy modifies the second fee to increase the slope of the linear fee function if the payment contributes to the balance in a negative manner. More concretely, they split a transaction amount val into $val_1$ and $val_2$. If the channel balance in the direction of routing is below half the total capacity, $val_1=0$.  If the balance minus the transaction value is at least half the total capacity, $val_2=0$. Otherwise, $val_1$ corresponds to the difference between the balance and and half the total capacity while $val_2$ is the remaining amount. The total fee is then $b + val_1 \cdot r_1 + val_2 \cdot r_2$ for rates $r_2 > r_1$. 
Our evaluation indicates that The Merchant leads to improved results in terms of maintaining a high success ratio.

\section{Conclusion}

This work illustrates how fees can incentivize maintaining balanced channels. However, they can not completely avoid depletion. Thus, it is important to combine them with other approaches such as on-chain rebalancing and loops. 
Future work should focus on coupling different strategies for rebalancing in payment channel networks.

\bibliographystyle{plain}
\bibliography{rebalance}

\end{document}